\documentclass{article}
\usepackage{tikz, subcaption, mathdots, amssymb, amsmath, braket, amsthm, cite}
\usepackage[margin = 2cm]{geometry}
\usepackage[colorlinks = true]{hyperref}

\newtheorem{lemma}{Lemma}
\newtheorem{definition}{Definition}
\newtheorem{procedure}{Procedure}

\DeclareMathOperator{\diag}{diag}

\title{Quantum routing in planar graph using perfect state transfer}
\author{Supriyo Dutta \\ Depertment of Mathematics \\ National Institute of Technology Agartala \\ Barjala, Jirania, Tripura(W), India - 799046. \\ Email: \texttt{dosupriyo@gmail.com}}
\date{} 

\begin{document}
	\maketitle 
	\begin{abstract}
		In this article, we consider a spin-spin interaction network governed by $XX + YY$ Hamiltonian. The vertices and edges of the network represent the spin objects and their interactions, respectively. We take a privilege to switch on or off any interaction, that assists us to perform multiple perfect state transfers in a graph simultaneously. We also build up a salable network allowing quantum communication between two arbitrary vertices. Later we utilize the combinatorial characteristics of hypercube graphs to propose a static routing schema to communicate simultaneously between a set of senders and a set of receivers in a planar network. Our construction is new and significantly powerful. We elaborate multiple examples of planar graphs supporting quantum routing where classical routing is not possible.
		
	\end{abstract}

	\section{Introduction}
		
		In the theory of quantum information processing \cite{benenti2004principles, pavicic2006quantum, nikolopoulos2014quantum}, we perform different communication tasks using quantum networks, for example, quantum state transfer, routing, switching and splitting of quantum signals, etc \cite{van2014quantum, kimble2008quantum, caleffi2018quantum, pirandola2016physics, cacciapuoti2019quantum}. We encode the information to be communicated onto quantum states. The elementary unit of quantum information is the qubit, which is a well-defined physical system to store the logical $0$ and $1$ as well as their arbitrary superpositions. A quantum network is a combination of a number of physical objects generating the qubits. The vertices of the network correspond to the qubits. We built up a larger network by linking a number of nodes together. Similarly, a quantum computing device consists of qubits, which are essentially static entities. We link distinct quantum processors effectively to transfer quantum information with high fidelity between different processors of quantum computing hardware. It supports the DiVincenzo criteria for scalable quantum architecture \cite{divincenzo2000physical}. 
		
		Instead of moving the qubits physically between different parts of a quantum architecture, we can realize buses for quantum states to travel between different locations. The quantum state will be transferred from qubit to qubit due to the interactions between spin objects, which neglects the necessity of moving entities. An example of such a system is a spin chain that consists of many permanently coupled qubits. In this article, we will use the terms qubits and spins interchangeably. Therefore, we may use a spin chain of static qubits to transport arbitrary quantum states from one place to another \cite{bose2003quantum, christandl2004perfect}. This strategy enables us to fabricate ``all solid-state” chips containing only a single species of qubits for information processing and transport. The interaction between spins is modelled by different types of spin-spin interaction Hamiltonians. The dynamics of $XX + YY$ Hamiltonian are utilized to connect remote registers of a scalable quantum architecture, in this work. This spin chain architecture provides the idea of data buses and entanglers between two extreme ends of the spin chain \cite{bose2014spin}.
		
		Another motivation behind the present work comes from superconducting quantum devices. The superconducting qubits can be treated as the vertices in a network. The edges correspond to microwave buses joining the qubits. The effective Hamiltonian \cite{majer2007coupling} representing two qubits of frequency $\omega_{1,2}$ interacting with a transmission line resonator with resonance frequency $\omega_r$, modeling the cavity,  is given by
		\begin{equation}\label{equn_1}
			H_{eff} = \frac{\hbar \omega_1}{2} \sigma_1^z + \frac{\hbar \omega_2}{2} \sigma_2^z + \hbar(\omega_r + \xi_1 \sigma_1^z + \xi_2 \sigma_2^z)a^\dagger a + \hbar J (\sigma_1^- \sigma_2^+ + \sigma_2^- \sigma_1^+).
		\end{equation} 
		The above Hamiltonian is obtained for the qubits strongly detuned from the resonator and after an adiabatically elimination of the resonant Jaynes-Cummings
		interaction.  Here $\xi_{1,2}$ are the frequency shifts which can be calculated from the detunings and the coupling strength of the resonator to the qubits. This can be generalized for an arbitrary number of qubits coupled to the same mode of the cavity resonator. The last term in the above Hamiltonian
		$\hbar J (\sigma_1^- \sigma_2^+ + \sigma_2^- \sigma_1^+)$ is the flip-flop interaction through virtual interaction with the resonator. Note that, this  term can be converted into an $XX + YY$ Hamiltonian.
		
		In this work, we first propose a method to construct a scalable quantum network, where we can communicate between two arbitrary vertices using Perfect State Transfer (PST). Then, we build up a quantum routing schema. In classical communication, routing is a process for finding paths in a network to communicate between a set of senders and a set of receivers \cite{huitema1995routing}. There have been many attempts to utilize the idea of state transfer in quantum routing \cite{bose2009communication, sadlier2019near, chudzicki2010parallel, pemberton2011perfect, zhan2014perfect, paganelli2013routing}. These proposals allow state transfer between two vertices of a network with a probability of less than $1$. Our idea is significantly different, in this context. We apply PST in every step of our communication protocol. Moreover, all of our networks are planar. The motivation behind considering planar graphs is building a quantum architecture which can be designed on a planar chip. The senders and receivers are the nodes at the infinite face of the network. The senders transmit quantum information to the receivers using a number of PST via a sequence of subgraphs. In classical routing, we find out a number of edge-disjoint paths between the senders and the receivers. This is not essential in the case of quantum routing. One subgraph may support communication between two pairs of senders and receivers, which may be considered an advantage of quantum communication.
		
		This article is distributed as follows. The next section describes a number of preliminary ideas on quantum state transfer and graph theory. In section 3, we establish a mathematical foundation which makes multiple state transfer in a network possible. In section 4, we construct a number of scalable quantum communication networks where we can communicate between any two vertices. Section 5 is dedicated to quantum routing. We propose a static routing protocol on arbitrary planar graphs.

	\section{Preliminary ideas}
	
		A graph $G = (V(G), E(G))$ is a combination of a vertex set $V(G)$ and an edge set $E(G) \subset V(G) \times V(G)$ \cite{west2001introduction}. We represent a spin-spin interaction with a graph. The vertices correspond to the qubits or spin objects. If two spins are allowed to interact, we join the corresponding vertices by exactly one edge and vice versa. In this article, we consider the interaction Hamiltonian as
		\begin{equation}
		H_{XY} = \frac{1}{2} \sum_{(u, v) \in E(G)} J_{u, v} ( \sigma^u_x \sigma^v_x + \sigma^u_y \sigma^v_y ).
		\end{equation} 
	    Here, $J_{u, v}$ is the coupling strength between two spin objects located at the vertices $u$ and $v$. Also, $\sigma^u_p = I_2 \otimes I_2 \otimes \dots \sigma_p (u\text{-th position})\otimes \dots \otimes I_2$, where $\sigma_p = \sigma_x$ and $\sigma_y$, are the Pauli $x$ and $y$ operators, respectively. Here, we assume that all possible coupling strength are equal to $1$, which makes $G$ a simple graph. 
		
		The adjacency matrix of a graph $G$ with $n$ vertices is defined by $A(G) = (a_{i,j})_{n \times n}$, where $a_{i,j} = 1$ if $(i, j) \in E(G)$ and ; $a_{i,j} = 0$ otherwise. It is proved that action of $H_{XY}$ on $\mathbb{C}^{2^n}$ is equivalent to the action of $\exp(-\iota A(G) t)$ on $\mathbb{C}^n$ when $J_{u, v} = 1$ for all $(u, v) \in E(G)$ \cite{osborne2006statics}. Now, corresponding to every vertex $u \in V(G)$ we assign a basis vector of $\mathbb{C}^n$, which is denoted by $\ket{u}$. There is a Perfect State Transfer (PST) between the vertices $u$ and $v$ at time $t = \tau$ if $\braket{v | \exp(-\iota A(G) \tau) | u} = 1$. For PST $u$ and $v$ must belong to the same connected component. PST is a rare phenomenon \cite{godsil2010can, cheung2011perfect, coutinho2015no}. There is no PST between two vertices if $G$ contains no edge.
		
		As mentioned above, we will limit ourselves to planar graphs in this work. We can embed a planar graph on a plane in such a way that its edges intersect only at their endpoints. When a planar graph is depicted without edges crossing, the vertices and edges divide the plane into regions or faces. The infinite face surrounds all other regions. A number of well-known planar graphs allowing PST are depicted in the figure \ref{graphs_with_PST}. They are path graphs with $2$ and $3$ vertices as well as the hypercube graphs of $4$ and $8$ vertices \cite{kendon2011perfect}. Note that, a hypercube graph with more than 8 vertices is not a planar graph. 
		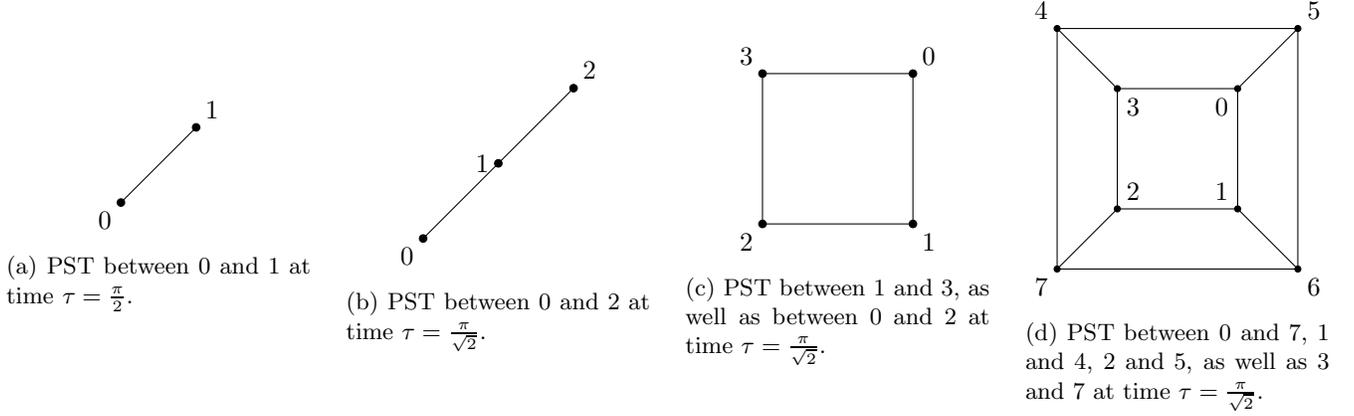
\begin{figure}
			\centering 
			\begin{subfigure}{.23 \textwidth}
				\centering
				\begin{tikzpicture}
					\draw [fill] (0, 0) circle [radius = .5mm];
					\node [below left] at (0, 0) {$0$};
					\draw [fill] (1, 1) circle [radius = .5mm];
					\node [above right] at (1,1) {$1$};
					\draw (0, 0) -- (1, 1);
				\end{tikzpicture}
				\caption{PST between $0$ and $1$ at time $\tau = \frac{\pi}{2}$.}
			\end{subfigure}
			\quad 
			\begin{subfigure}{.23 \textwidth}
				\centering
				\begin{tikzpicture}
					\draw [fill] (0, 0) circle [radius = .5mm];
					\node [below left] at (0, 0) {$0$};
					\draw [fill] (1, 1) circle [radius = .5mm];
					\node [left] at (1, 1) {$1$};
					\draw [fill] (2, 2) circle [radius = .5mm];
					\node [above right] at (2, 2) {$2$};
					\draw (0, 0) -- (2, 2);
				\end{tikzpicture}
				\caption{PST between $0$ and $2$ at time $\tau = \frac{\pi}{\sqrt{2}}$.}
			\end{subfigure}
			\quad 
			\begin{subfigure}{.23 \textwidth}
				\centering
				\begin{tikzpicture}
					\draw [fill] (1, 1) circle [radius = .5mm];
					\node [above right] at (1, 1) {$0$};
					\draw [fill] (1, -1) circle [radius = .5mm];
					\node [below right] at (1, -1) {$1$};
					\draw [fill] (-1, -1) circle [radius = .5mm];
					\node [below left] at (-1, -1) {$2$};
					\draw [fill] (-1, 1) circle [radius = .5mm];
					\node [above left] at (-1, 1) {$3$};
					\draw (1, 1) -- (1, -1) -- (-1, -1) -- (-1, 1) -- (1, 1);
				\end{tikzpicture}
				\caption{PST between $1$ and $3$, as well as between $0$ and $2$ at time $\tau = \frac{\pi}{\sqrt{2}}$.}
			\end{subfigure}
			\quad 
			\begin{subfigure}{.23 \textwidth}
				\centering
				\begin{tikzpicture}[scale = .8]
					\draw [fill] (1, 1) circle [radius = .5mm];
					\node [below left] at (1, 1) {$0$};
					\draw [fill] (1, -1) circle [radius = .5mm];
					\node [above left] at (1, -1) {$1$};
					\draw [fill] (-1, -1) circle [radius = .5mm];
					\node [above right] at (-1, -1) {$2$};
					\draw [fill] (-1, 1) circle [radius = .5mm];
					\node [below right] at (-1, 1) {$3$};
					\draw (1, 1) -- (1, -1) -- (-1, -1) -- (-1, 1) -- (1, 1);
					\draw [fill] (2, 2) circle [radius = .5mm];
					\node [above right] at (2, 2) {$5$};
					\draw [fill] (2, -2) circle [radius = .5mm];
					\node [below right] at (2, -2) {$6$};
					\draw [fill] (-2, -2) circle [radius = .5mm];
					\node [below left] at (-2, -2) {$7$};
					\draw [fill] (-2, 2) circle [radius = .5mm];
					\node [above left] at (-2, 2) {$4$};
					\draw (2, 2) -- (2, -2) -- (-2, -2) -- (-2, 2) -- (2, 2);
					\draw (1, 1) -- (2, 2);
					\draw (1, -1) -- (2, -2);
					\draw (-1, 1) -- (-2, 2);
					\draw (-1, -1) -- (-2, -2);
				\end{tikzpicture}
				\caption{PST between $0$ and $7$, $1$ and $4$, $2$ and $5$, as well as $3$ and $7$ at time $\tau = \frac{\pi}{\sqrt{2}}$.}
				\label{hypercube_with_8_vertices}
			\end{subfigure}
			\caption{Examples of well-known planar graphs with PST \cite{christandl2004perfect}.}
			\label{graphs_with_PST}
		\end{figure}
		
		Given two graphs $G_1 = (V(G_1), E(G_1))$ and $G_2 = (V(G_2), E(G_2))$ we write $G_1 = G_2$ if $V(G_1) = V(G_2)$, $E(G_1) = E(G_2)$ and they have same vertex labeling. A graph $G_1$ is said to be a subgraph of a graph $G_2$ if $V(G_1) \subset V(G_2)$ and $E(G_1) \subset E(G_2)$. The union of two graphs $G_1$ and $G_2$ is a new graph  $G = G_1 \cup G_2 = (V(G), E(G))$, such that, $V(G) = V(G_1) \cup V(G_2)$ and $E(G) = E(G_1) \cup E(G_2)$. Note that, $G_1$ and $G_2$ are subgraphs of $G_1 \cup G_2$. In a similar fashion we can define $G = G_1 \cup G_2 \cup \dots G_m$ for $m$ subgraphs $G_1, G_2, \dots G_m$. If $A(G_i)$ is the adjacency matrix of $G_i$, then the adjacency matrix of $G$ is given by a block diagonal matrix $A(G) = \diag\{A(G_1), A(G_2), \dots , A(G_m)\}$. It is easy to prove that 
		\begin{equation}
			\exp(-\iota A(G) \tau) = \diag\{\exp(-\iota A(G_1) \tau), \exp(-\iota A(G_2) \tau), \dots , \exp(-\iota A(G_m) \tau)\}.
		\end{equation}
		
		A path of length $p$ is a sequence of distinct vertices and edges $\{v_1, e_1, v_2, e_2, \dots e_p, v_{p + 1}\}$, such that $e_i = (v_i, v_{i + 1})$ for $i = 1, 2, \dots p$. Two paths are said to be edge-disjoint if there is no edge belonging to both the paths. Two subgraphs $G_1$ and $G_2$ are said to be disjoint if $V(G_1) \cap V(G_2) = \emptyset$ and $E(G_1) \cap E(G_2) = \emptyset$. A graph is said to be connected if there is a path between any two vertices. Connectivity between two vertices is an equivalence relation. The equivalence classes are called the connected components of a graph. The distance between two vertices $u$ and $v$ is denoted by $d(u, v)$ and defined by the length of the shortest path between them. The diameter of a graph is $d = \max_{u, v \in V(G)} d(u, v)$, that is the maximum distance between any two vertices in it. The eccentricity $\epsilon(v)$ of a vertex $v$ is the greatest distance between $v$ and any other vertex in the graph. Symbolically, $\epsilon (v) = \max _{u\in V(G)} d(v,u)$. The radius of a graph is $r$, which is the minimum eccentricity of the vertices in $G$. A central vertex in a graph of radius $r$ is a vertex with eccentricity $r$. The set of all central vertices in $G$ is denoted by $C(G)$, which is called the centre of $G$.

	\section{Hamiltonian engineering}
		
		Hamiltonian engineering was utilized in NMR for achieving NMR diffraction in solid \cite{waugh1968approach, mansfield1975diffraction}. We turn on couplings between a number of selective excitations \cite{warren1979selective} and decouple then dynamically. We can accomplish simultaneous tunability of the coupling strengths by exploiting the magnetic-field gradients. The goals of Hamiltonian engineering is summarized as follows \cite{ajoy2013quantum, singh2020perfect}:
		\begin{enumerate}
			\item 
			The cancellation of unwanted couplings, that is $J_{u, v} = 0$ if the edge $(u, v)$ is not involved in PST. 
			\item
			Performing a PST with the remaining couplings with $J_{u, v} = 1$.
		\end{enumerate}
		In this work, we compare the coupling between qubits with the edges between the nodes. Therefore, switching on or off a particular edge is equivalent to coupling and decoupling the interactions. Here, we mathematically justify how PST holds in a graph after switching on or off a number of edges. 
		
		\begin{lemma}\label{single_PST}
			Let $G$ be a graph with $n$ vertices allowing PST between $u$ and $v$ at time $\tau$. Then there is a state transfer between $u$ and $v$ at time $\tau$ in the graph $G' = G \cup {v_1}  \cup {v_2} \cup \dots  \cup {v_m}$, where $v_1, v_2, \dots v_m$ are isolated vertices of any number $m$.
		\end{lemma}		
	
		\begin{proof}
			As there is a PST between $u$ and $v$ at time $t = \tau$, there are vectors $\ket{u} = [0, 0, \dots 1(u\text{-th position}), \dots 0]^\dagger$ and $\ket{v} = [0, 0, \dots 1(v\text{-th position}), \dots 0]^\dagger$ in $\mathbb{C}^n$, such that, $\braket{v | \exp(-\iota A(G) \tau) | u} = 1$. Note that, $A(G')$ is a block diagonal matrix given by $A(G') = \diag\{A(G), 0, 0, \dots 0 (m\text{- times})\}$. Corresponding to the vertices $u$ and $v$ define new vectors $\ket{U}$ and $\ket{V}$ in $\mathbb{C}^{n + m}$ such that $\ket{U} = [\bra{u}, 0, 0, \dots 0 (m\text{- times})]^\dagger$, and $\ket{V} = [\bra{v}, 0, 0, \dots 0 (m\text{- times})]^\dagger$. Now,
			\begin{equation}
				\begin{split}
					& \exp(-\iota A(G') \tau) = \diag\{\exp(-\iota A(G) \tau), -\iota \tau, -\iota \tau, \dots -\iota\tau(m\text{- times})\} \\
					\text{or}~ & \braket{V | \exp(-\iota A(G') \tau)| U } = \braket{V | \diag\{\exp(-\iota A(G) \tau), -\iota \tau, -\iota \tau, \dots -\iota\tau(m\text{- times})\} | U } \\
					& \hspace{3.4cm} = \braket{v | \exp(-\iota A(G) \tau) | u} -\iota \tau \times 0 -\iota \tau \times 0 \dots -\iota\tau \times 0(m\text{- times}) = 1
				\end{split}
			\end{equation}
			Therefore, there is a state transfer from vertex $u$ to $v$ in the new graph $G'$.
		\end{proof}
		
		\begin{lemma}\label{multi_PST}
			Let $G_1, G_2, \dots G_m$ be the graphs allowing PSTs between the vertices $u_i$ and $v_i$ at time $\tau_i$, where $u_i$ and $v_i \in V(G_i)$ for $i = 1, 2, \dots m$. Then, there are PSTs between $u_i$ and $v_i$ at time $\tau_i$ in the graph $G = G_1 \cup G_2 \cup \dots \cup G_m$.
		\end{lemma}		
		
		\begin{proof}
			Let the graph $G_i$ has $n_i$ vertices for $i = 1, 2, \dots m$. Then $G$ has $n = n_1 + n_2 + \dots n_m$ vertices. Corresponding to the vertices $u_i$ and $v_i$ we define the vectors of dimension $n_i$ such that $\ket{u_i} = [0, 0, \dots 1(u_i\text{-th position}), \dots 0]^\dagger$ and $\ket{v_i} = [0, 0, \dots 1(v_i\text{-th position}), \dots 0]^\dagger \in \mathbb{C}^{n_i}$. As there is a PST from $u_i$ to $v_i$ we can write $\braket{v_i | \exp(-\iota A(G_i) \tau_i) | u_i} = 1$. Corresponding to a vector $\ket{u_i} \in \mathbb{C}^{n_i}$ construct a new vector $\ket{U_i} \in \mathbb{C}^{n}$ where $\ket{U_i} = [\bra{0_1}, \bra{0_2}, \dots \bra{u_i}, \dots \bra{0_m}]^\dagger$, where $\bra{0_j}$ is all zero vector of dimension $n_j$. Now
			\begin{equation}
				\begin{split}
					& \braket{V_i | \exp(-\iota A(G) \tau_i) | U_i} = \braket{V_i | \exp(-\iota \diag\{A(G_1), A(G_2), \dots A(G_i) \dots A(G_m)\} \tau_i) | U_i} \\
					& = \braket{V_i | \diag\{\exp(-\iota A(G_1)\tau_i), \exp(-\iota A(G_2)\tau_i), \dots, \exp(-\iota A(G_i)\tau_i), \dots \exp(-\iota A(G_m)\tau_i) \} | U_i} \\
					& = \braket{0_1 | \exp(-\iota A(G_1)\tau_i) | 0_1} + \braket{0_2 | \exp(-\iota A(G_2)\tau_i) | 0_2} + \dots + \braket{v_i | \exp(-\iota A(G_i)\tau_i) | u_i} \\
					& \hspace{5cm} + \dots + \braket{0_m | \exp(-\iota A(G_m)\tau_i) | 0_m} \\
					& = 0 + 0 + \dots + 1 + \dots + 0 = 1. 
				\end{split}
			\end{equation}
			Therefore, there is a PST between vertices $u_i$ and $v_i$ in the graph $G$.
		\end{proof}
	
		Lemma \ref{multi_PST} suggests that perfect state transfer is possible simultaneously between two vertices in a connected component of $G$. Let $G_1, G_2, \dots G_m$ be the graphs allowing PSTs between the vertices $u_i$ and $v_i \in V(G_i)$ at time $\tau$ for all $i = 1, 2, \dots m$. Then, there are simultaneous PSTs between $u_i$ and $v_i$ at time $\tau$ in the graph $G = G_1 \cup G_2 \cup \dots \cup G_m$. Therefore, switching off a number of interactions we can generate multiple connected components allowing PSTs between vertices simultaneously. This process modifies the interaction Hamiltonian. 
		
		It is practically difficult to keep the coupling strengths equal throughout the PST process. Therefore, the state transfer cost is proportional to the number of edges in the subgraph. To reduce the state transfer cost, we should construct the subgraphs with the minimum number of edges. Among the graphs performing PST at time $\frac{\pi}{\sqrt{2}}$ at a distance $2$ depicted in the the figure \ref{graphs_with_PST} path graph of length $2$ has minimum number of edges. 
		
		Now we discuss how to use Hamiltonian engineering to transfer information from a sender to a receiver located at the vertices $s$ and $r$, respectively. For simplicity, we assume that this switching operation on the edges is not time consuming. Consider a sequence of vertices $s = s_0, s_1, s_2, \dots s_p = r$ and a sequence of subgraphs $G_0, G_1, \dots G_{p - 1}$ such that $s_0 \in V(G_0)$; $s_j \in V(G_{j - 1}) \cap V(G_{j})$ for $j = 1, 2, \dots (p - 1)$; and $s_p \in G_{p - 1}$. In addition, $G_j$ allows PST from $s_j$ to $s_{j + 1}$ for $j = 0, 1, 2, \dots {p - 1}$.  During PST in the subgraph $G_j$ all the edges in $E(G) - E(G_j)$ will remain switched off. This operation converts the graph $G$ to another graph $G'$ which is $G_j$ union a number of isolated vertices. Lemma \ref{single_PST} suggests that a PST is possible in $G'$ from $s_{j}$ to $s_{j + 1}$. In the next step, we switched off the edges in $G_j$ and switched on the edges in $G_{j + 1}$, which makes the next PST possible. Therefore, we need $p$ number of PSTs to transfer the information from $s$ to $r$ using the sequence of subgraphs $\{G_j\}$. We describe this process with an example in figure \ref{example_state_transition}. These ideas influence the development of quantum routing and procedures for generating scalable quantum communication network, which we shall discuss in the next two sections. 
		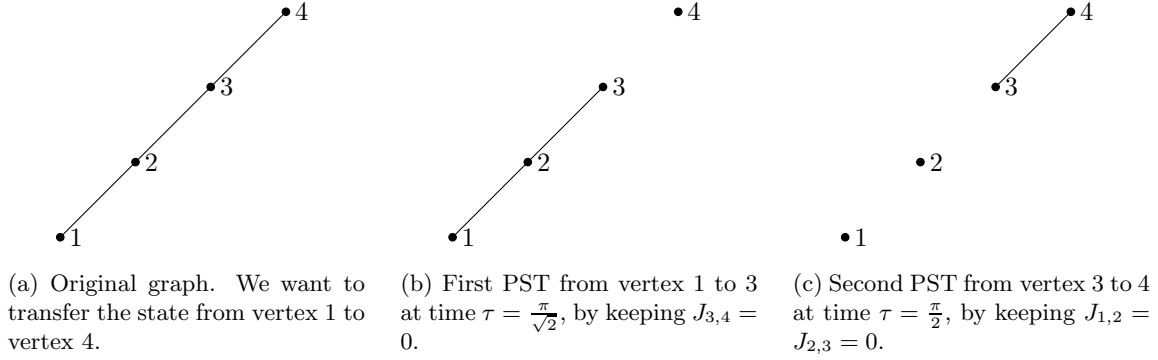
\begin{figure}
			\centering 
			\begin{subfigure}{.27 \textwidth}
				\centering 
				\begin{tikzpicture}
				\draw [fill] (0, 0) circle [radius = .5mm];
				\node [right] at (0, 0) {$1$};
				\draw [fill] (1, 1) circle [radius = .5mm];
				\node [right] at (1, 1) {$2$};
				\draw [fill] (2, 2) circle [radius = .5mm];
				\node [right] at (2, 2) {$3$};
				\draw [fill] (3, 3) circle [radius = .5mm];
				\node [right] at (3, 3) {$4$};
				\draw (0, 0) -- (3, 3);
				\end{tikzpicture}
				\caption{Original graph. We want to transfer the state from vertex $1$ to vertex $4$.}
			\end{subfigure}
			\quad 
			\begin{subfigure}{.27 \textwidth}
				\centering 
				\begin{tikzpicture}
				\draw [fill] (0, 0) circle [radius = .5mm];
				\node [right] at (0, 0) {$1$};
				\draw [fill] (1, 1) circle [radius = .5mm];
				\node [right] at (1, 1) {$2$};
				\draw [fill] (2, 2) circle [radius = .5mm];
				\node [right] at (2, 2) {$3$};
				\draw [fill] (3, 3) circle [radius = .5mm];
				\node [right] at (3, 3) {$4$};
				\draw (0, 0) -- (2, 2);
				\end{tikzpicture}
				\caption{First PST from vertex $1$ to $3$ at time $\tau = \frac{\pi}{\sqrt{2}}$, by keeping $J_{3, 4} = 0$.}
			\end{subfigure}
			\quad 
			\begin{subfigure}{.27 \textwidth}
				\centering 
				\begin{tikzpicture}
				\draw [fill] (0, 0) circle [radius = .5mm];
				\node [right] at (0, 0) {$1$};
				\draw [fill] (1, 1) circle [radius = .5mm];
				\node [right] at (1, 1) {$2$};
				\draw [fill] (2, 2) circle [radius = .5mm];
				\node [right] at (2, 2) {$3$};
				\draw [fill] (3, 3) circle [radius = .5mm];
				\node [right] at (3, 3) {$4$};
				\draw (2, 2) -- (3, 3);
				\end{tikzpicture}
				\caption{Second PST from vertex $3$ to $4$ at time $\tau = \frac{\pi}{2}$, by keeping $J_{1, 2} = J_{2, 3} = 0$.}
			\end{subfigure}
			\caption{A path graph of length $4$ does not support PST between vertices $1$ and $4$. But, we can communication between them using two PST.}
			\label{example_state_transition} 
		\end{figure}

	\section{Construction of scalable network allowing communication between any two vertices}
		
		Our objective for constructing a network is twofold. We place a sender and a receiver at two distinct vertices of a graph. First, we want to enable them to communicate using a sequence of PSTs. Next, we can increase the number of nodes in the graph to make the network scalable. 
		
		\begin{definition}
			We define a graph as $p$-PST if we can communicate between any two vertices in it using at most $p$ number of perfect state transfers.
		\end{definition}
		A few examples of $1$-PST graphs are available in figure \ref{graphs_with_PST}. They need no switching operation for information transfer. All complete graphs, complete bipartite graphs, star graphs, wheel graphs, friendship graphs, and Petersen graphs are $1$-PST. Examples of $2$ PST graphs are path graphs with $4$ or $5$ vertices and cycle graphs with $6$ to $9$ vertices. The following lemma relates the diameter of the graph and the number of required PSTs to communicate between any two vertices.

		\begin{lemma}
			Any planar graph with a diameter at most $2p$ is $p$-PST. Also, the state transfer time between any two vertices representing a sender and a receiver is at most $\frac{p \pi}{\sqrt{2}}$.
		\end{lemma} 
	
		\begin{proof}
			The diameter of the graph is $2p$. Therefore, three are at least two vertices $u$ and $v$, such that $d(u, v) = 2p$ and a path consists of the vertices  $u = u_0, u_1, u_2, \dots u_{2p} = v$. Now break the path from $u$ to $v$ into paths of length $2$, which are $G_0 = \{u_0, u_1, u_2\}, G_1 = \{u_2, u_3, u_4\}, \dots, G_{p - 1} = \{u_{2p - 2}, u_{2p - 1}, u_{2p}\}$. Using the graph $G_i$ we perform a PST from $u_{2i}$ to $u_{2i + 2}$ at time $\tau = \frac{\pi}{\sqrt{2}}$ by keeping all edges in $E(G) - E(G_i)$ switched off. In this way, we transfer the state from $u$ to $v$ using at most $p$ PSTs at time $\frac{p \pi}{\sqrt{2}}$.
		\end{proof}
			
		The above lemma suggests that if the planar graph $G$ has diameter $d$ then we need at most $\frac{d}{2}$ PST if $d$ is an even number. Here, state transfer time is $\frac{d \pi}{\sqrt{2}}$. If $d$ is an odd number then we need at most $\left[\frac{d}{2} \right] + 1$ PSTs. The state transfer time is $\left[\frac{d}{2} \right]\frac{\pi}{\sqrt{2}} + \frac{\pi}{2}$, where $\left[\frac{d}{2} \right]$ denotes the greatest integer not greater than $\frac{d}{2}$. In this case, we do $\left[\frac{d}{2} \right]$ PSTs of length $2$. Each of them takes time $\frac{\pi}{\sqrt{2}}$ and $1$ PST with length $1$ that takes time $\frac{\pi}{2}$. 
	
		Now, we describe two procedures to build up an infinitely scalable network supporting PST between any two vertices. We start the process with an initial graph $G_0$ with a fixed planar embedding. From $G_0$ we create an infinite sequence of planar graphs $G_i$ supporting quantum communication between any two of its vertices. The construction of infinite families of planar graphs with fixed diameter is studied in \cite{hell1993largest}. To discuss the first procedure we introduce the idea of the glued graph, which is a new graph $G$ generated from two given graphs $G_1$ and $G_2$ by glueing a vertex of $G_1$ with a vertex of $G_2$ \cite{lovasz2006graph, contreras2020gluing}.
		\begin{definition}
			\textbf{Glued graph}: Let $G_1$ and $G_2$ be two graphs such that $V(G_1) = \{v_{1,1}, v_{1,2}, \dots v_{1,n_1}\}$ and $V(G_2) = \{v_{2,1}, v_{2,2}, \dots v_{2,n_2}\}$. We glue $v_{1,n_1}$ and $v_{2,1}$ two build up a new graph $G$ with $V(G) = \{v_i: i = 1, 2, \dots (n_1 + n_2 - 1)\}$ where $v_i = v_{1, i}$ for $1 \leq i \leq n_1$ and $v_i = v_{2, i - n_1 + 1}$ where $n_1 \leq i \leq (n_1 + n_2 -1)$. Also, $E(G) = \{(v_i, v_j): (v_{1, i}, v_{1, j}) \in E(G_1) ~\text{and}~ 1 \leq i, j \leq n_1 \} \cup \{(v_i, v_j): (v_{2, i - n_1 + 1}, v_{2, j - n_1 + 1}) \in E(G_2) ~\text{and}~ n_1 \leq i, j \leq (n_1 + n_2 - 1) \}$. We denote $G = G_1 \curlywedge^{1,n_1}_{2,1} G_2$.
		\end{definition}
		Consider planar graphs $G_1$ and $G_2$, such that $v_{1, n_1} \in C(G_1)$ and $v_{2,1} \in C(G_2)$. It is easy to observe that if $v_{2,1}$ lies at the boundary of the infinite face of $G_2$, then the glued graph $G_1 \curlywedge^{1,n_1}_{2,1} G_2$ is a planar graph. Let $r_i$ and $d_i$ be the radius and diameter of the graphs $G_i$, respectively for $i = 1$ and $2$. Then, the diameter of $G_1 \curlywedge^{1,n_1}_{2,1} G_2$ will be $\max\{d_1, d_2, r_1 + r_2\}$.
		
		\begin{procedure}\label{procedure1}
			Let the initial graph $G_0$ has $n_0$ vertices and $v_{0, n_0} \in C(G_0)$. Consider the planar graphs $G_i$ such that $v_{i, 1} \in C(G_i)$ and $v_{i, 1}$ lies at the boundary of infinite face of $G_i$ for $i = 1, 2, \dots $. Define $G_0' = G_0$. In the $i$-th iteration generate a new graph $G_i' = G_{i - 1}' \curlywedge^{0,n_0}_{i,1} G_i$ for $i = 1, 2, 3, \dots$.
		\end{procedure}
		
		In procedure \ref{procedure1}, if $G_i$ contains $n_i$ vertices for $i = 0, 1, 2, \dots$, then number of vertices in $G_k'$ will be $\sum_{i = 0}^k n_j - k$. Clearly, $G_i$ is a subgraph of $G_k'$ for all $i = 0, 1, 2, \dots k$. Let the radius and diameter of $G_i$ are $r_i$ and $d_i$, respectively. Consider any two vertices $u$ and $v$ of $G_k'$. If $u$ and $v$ in $V(G_i)$ for some $i$, then we can communicate between them using at most $\left[\frac{d_i}{2}\right] + 1$ PSTs. If $u$ and $v$ belong to $G_i$ and $G_j$, respectively, for $i \neq j$, then we can construct a path between $u$ and $v$ via $v_{0,n_0}$. Length of the path is at most $r_i + r_j$, as $v_{0,n_0}$ is a central vertex of both $G_i$ and $G_j$. The number of required PSTs is $\left[\frac{r_i + r_j}{2}\right] + 1$.
		
		For example, let $G_0$ be a graph with only one vertex $v_{0, 1}$ and $G_i = C_4$, the cycle graph with $4$ vertices for $i = 1, 2, \dots$. Note that, the single vertex in $G_0$ is trivially a central vertex. Any vertex of $C_4$ is a central vertex and all of them lies at the boundary of the infinite face. We glue a vertex of $C_4$ with $v_{0, 1}$ in every step of the iteration. The new graph is depicted in figure \ref{example_of_proedure_1}. Here, we need at most two PSTs to transfer information between any two vertices. Also, the maximum state transfer time is $\sqrt{2} \pi$.
		\begin{figure}
			\begin{subfigure}{.45\textwidth}
				\centering
				\begin{tikzpicture}
					\draw [fill] (0, 0) circle [radius = .5mm];
					\node [above right] at (0, 0) {$v_{0, 1}$};
					\draw [fill] (2, 0) circle [radius = .5mm];
					\draw [fill] (1.5, .5) circle [radius = .5mm];
					\draw [fill] (1.5, -.5) circle [radius = .5mm];
					\draw (0,0) -- (1.5, .5) -- (2, 0) -- (1.5, -.5) -- (0, 0);
					\draw [fill] (-2, 0) circle [radius = .5mm];
					\draw [fill] (-1.5, .5) circle [radius = .5mm];
					\draw [fill] (-1.5, -.5) circle [radius = .5mm];
					\draw (0,0) -- (-1.5, .5) -- (-2, 0) -- (-1.5, -.5) -- (0, 0);
					\draw [fill] (0, 2) circle [radius = .5mm];
					\draw [fill] (.5, 1.5) circle [radius = .5mm];
					\draw [fill] (-.5, 1.5) circle [radius = .5mm];
					\draw (0, 0) -- (.5, 1.5) -- (0, 2) -- (-.5, 1.5) -- (0, 0);
					\draw [fill] (0, -2) circle [radius = .5mm];
					\draw [fill] (.5, -1.5) circle [radius = .5mm];
					\draw [fill] (-.5, -1.5) circle [radius = .5mm];
					\draw (0, 0) -- (.5, -1.5) -- (0, -2) -- (-.5, -1.5) -- (0, 0);
				\end{tikzpicture}
				\caption{Example of procedure 1: The initial graph $G_0$ consists of a single vertex $a$. We add $C_4$ graphs with it. All the new graphs are $2$-PST.}
				\label{example_of_proedure_1} 
			\end{subfigure}
			\hspace{1cm}
			\begin{subfigure}{.45\textwidth}
				\centering 
				\begin{tikzpicture}
					\draw [fill] (0, 1) circle [radius = .5mm];
					\node [above] at (0, 1) {$u$};
					\draw [fill] (0, -1) circle [radius = .5mm];
					\node [below] at (0, -1) {$v$};
					\draw [fill] (1, 0) circle [radius = .5mm];
					\draw (0, 1) -- (1, 0) -- (0, -1);
					\draw [fill] (-1, 0) circle [radius = .5mm];
					\draw (0, 1) -- (-1, 0) -- (0, -1);
					\draw [fill] (2, 0) circle [radius = .5mm];
					\draw (0, 1) -- (2, 0) -- (0, -1);
					\draw [fill] (-2, 0) circle [radius = .5mm];
					\draw (0, 1) -- (-2, 0) -- (0, -1);
					\draw [fill] (3, 0) circle [radius = .5mm];
					\draw (0, 1) -- (3, 0) -- (0, -1);
					\draw [fill] (-3, 0) circle [radius = .5mm];
					\draw (0, 1) -- (-3, 0) -- (0, -1);
					\draw [fill] (4, 0) circle [radius = .5mm];
					\draw (0, 1) -- (4, 0) -- (0, -1);
					\draw [fill] (-4, 0) circle [radius = .5mm];
					\draw (0, 1) -- (-4, 0) -- (0, -1);
					\node [left] at (-4, 0) {$\dots$};
					\node [right] at (4, 0) {$\dots$};
				\end{tikzpicture}
				\caption{Example of procedure \ref{procedure2}: The graph $G_0 = C_4$. We add new vertices in the graph and join that with $u$ and $v$ to construct $G_1, G_2, \dots$. All the new graphs are $1$-PST.}
				\label{example_of_proedure_2}
			\end{subfigure}
			\caption{Construction of scalable planar graphs allowing communication using PSTs}
		\end{figure}
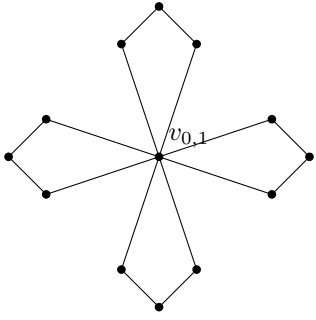
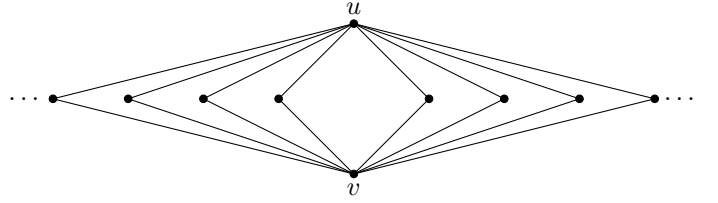
		
		This procedure increases the degree of vertex $v_{0, n_1}$ in every step which acts as a hub in the communication network. This may be considered a drawback of this procedure. Removing this vertex will generate a large number of disjoint subgraphs and which will reduce the efficiency of the communication network.
		
		\begin{procedure}\label{procedure2}
			Let $u$ and $v$ be any two vertices on the boundary of infinite face of the graph $G_0$. Construct new graphs $G_i = (V(G_i), E(G_i))$ where $V(G_i) = V(G_{i - 1}) \cup \{v_i\}$ and $E(G_i) = E(G_{i - 1}) \cup \{(u, v_i), (v, v_i)\}$ for $i = 1, 2, \dots$.
		\end{procedure}
	
		Procedure \ref{procedure2} is advantageous as compared to procedure \ref{procedure1}. If $d_i$ is diameter of the sequence of graphs $G_i$ generated by procedure \ref{procedure2} then we can prove $d_{i + 1} \geq d_i$ for $i = 0, 1, 2, \dots$. It reduces the number of PSTs when the size of network increases and $d_{i + 1} > d_i$ for some $i$. Moreover, it generates two hubs $u$ and $v$. An example of procedure \ref{procedure2} is depicted in figure \ref{example_of_proedure_2}. Here, $G_0 = C_4$. In every step, we add new vertices with two particular vertices $u$ and $v$ in $G_0$.

	\section{Routing in arbitrary planar graph}
		
		In classical communication, a planar routing problem is defined on a planar graph $G$ with $n$ vertices and a fixed embedding in a plane \cite{frederickson1989efficient}. A net $N_i$ is a pair of vertices $(s_i, r_i)$ on the boundary of the infinite face of $G$. Here $s_i$ and $r_i$ denotes the $i$-th sender and receiver in the graph. The set containing all considered nets is $Ne = \{N_i: i = 1, 2, \dots q\}$. The classical routing problem is to find a set of pairwise edge-disjoint paths $P_i$ corresponding to the net $N_i \in Ne$, such that $P_i$ connects the two terminals $s_i$ and $r_i$ of $N_i$, simultaneously for all $i$ \cite{becker1986algorithms}. Here, the solver wants to find the shortest paths joining the two ends of the nets. In static routing, we can not add or remove any vertex or edge throughout the communication process. 
	
		We reformulate the routing problem to solve it in the domain of quantum communication. Here also, we cannot add or remove any vertex or an edge in the graph permanently. But, to perform a sequence of PSTs we keep a number of edges switched off in a particular time interval. Given a planar graph and a set of nets $Ne$ the quantum routing problem is to communicate between the terminals of the nets using PSTs. The solution of this problem contains a sequences of subgraphs $\{G_j\}_{j = 0, 1, \dots p}$ supporting PST with the following characteristics:
		\begin{enumerate}
			\item
				For every net corresponding to $i = 1, 2, \dots q$ there is a sequence of vertices $s_i = s_{i, 0}, s_{i, 1}, s_{i, 2}, \dots s_{i, p} = r_i$ and a sequence of subgraphs $G_0, G_1, \dots G_{p - 1}$ such that $s_{i, 0} \in G_0$; $s_{i, j} \in G_{j - 1} \cap G_j$ for $j = 1, 2, \dots (p - 1)$; and $s_{i, p} \in G_p$.
			\item
				The subgraph $G_j$ either allows PST from $s_{i, j}$ to $s_{i, j+1}$ or keeps the state unmoved at $s_{i, j}$.
			\item
				In $G_j$ we have $s_{i, j} \neq s_{k, j}$ and $s_{i, (j+1)} \neq s_{k, (j+1)}$. Note that $G_j$ allows PST from $s_{i, j}$ to $s_{i, (j+1)}$, and from $s_{k, j}$ to $s_{k, (j+1)}$ simultaneously. Either $s_{i, j} = s_{k, j}$, or $s_{i, (j+1)} = s_{k, (j+1)}$ will make PST impossible.
		\end{enumerate}
		We represent the quantum routing process with a routing table. Corresponding to every net we make a row where we register the sequence of vertices from sender to the receiver. We use the subgraphs $G_j$ to perform $j$-th PST, which correspond a column in the routing table. If the subgraph $G_j$ allows PST from $s_{i, j}$ to $s_{i, j+1}$, we denote $(s_{i, j}, s_{i, j+1})$ in $i$-th row and $j$-th column. If $G_j$ keeps the state at $s_{i, j}$ unmoved we mark $\{s_{i,j}\}$ in $i$-th row and $j$-th column. Note that $G_0$ and $G_p$ are the subgraphs without any edge.
		
		Now, we present the examples to demonstrate that quantum routing is more efficient than its classical counterpart. Consider the graph in figure \ref{routing with C4} with two nets $N_1 = (1, 5)$ and $N_2 = (6, 4)$. Note that, there are no edge-disjoint paths $P_1$ and $P_2$ joining the terminals points of $N_1$ and $N_2$, respectively. Therefore, simultaneous classical communication between the terminals of $N_1$ and $N_2$ is not possible for the graph in \ref{routing with C4}. We can transfer information using PST from $1$ to $2$ and $6$ to $7$ simultaneously keeping all edges switched off except $(1, 2)$ and $(6, 7)$. Then, we switched off the edges $(1, 2)$ and $(6, 7)$ and switched on the edges in the hypercube which is the induced subgraph generated by the vertices $\{2, 3, 7, 8\}$. Simultaneous PST is possible from $2$ to $8$ and $7$ to $3$, in this hypercube. Next, we switched off the hypercube and switched on the edges $(3, 4)$ and $(8, 5)$ to perform PST from $3$ to $4$ and from $8$ to $5$. Therefore, this graph supports simultaneous quantum communications between the terminals of $N_1$ and $N_2$. The quantum routing is described in the following routing table:
		\begin{center}
			\begin{tabular}{| c | c| c |c | c | c | c |}
				\hline
				& & \multicolumn{5}{|c|}{Subgraphs for PSTs} \\
				\hline 
				\multicolumn{2}{|c|}{Net} & $G_0$ & $G_1$ & $G_2$ & $G_3$ & $ G_4$ \\
				\hline
				$i = 1$ & $(1, 5)$ & $\{1\}$ & $(1, 2)$ & $(2, 8)$ & $(8, 5)$ & $\{5\}$ \\
				\hline
				$i = 2$ & $(6, 4)$ & $\{6\}$ & $(6, 7)$ & $(7, 3)$ & $(3, 4)$ & $\{4\}$\\
				\hline  
			\end{tabular}
		\end{center}
		Here, $G_1$ is a subgraph with two edges $(1, 2)$ and $(6, 7)$ as well as isolated vertices $3, 4, 8$ and $5$. Also, $G_2$ is the hypercube consists of four vertices $\{2, 3, 7, 8\}$ as well as four isolated vertices $1$, $4$, $5$ and $6$. For the final PST we use the subgraph $G_3$. It has edges $(3, 4)$ and $(8, 5)$ and isolated vertices $1$, $2$, $6$ and $7$. The initial and final graphs $G_0$ and $G_4$ have no edge in them.
		
		Similarly, in the graph of figure \ref{routing with H3} consider nets $N_1 = (8, 12), N_2 = (9, 13), N_3 = (15, 10)$ and $N_4 = (14, 11)$. Edge disjoint paths joining the terminals of these nets are impossible. Therefore, we can not perform classical communication between the terminals of these nets simultaneously. Hence, classical routing in planar graph is impossible for the graph in \ref{routing with H3}. But, we can do quantum routing similarly, as enunciated in the table below:
		\begin{center} 
			\begin{tabular}{| c | c | c | c | c | c| c | c | c |}
				\hline 
				& & \multicolumn{7}{|c|}{Subgraphs for PSTs} \\
				\hline 
				\multicolumn{2}{|c|}{Net} & $G_0$ & $G_1$ & $G_2$ & $G_3$ & $G_4$ & $G_5$ & $G_6$\\
				\hline 
				$i = 1$ & $(8, 12)$ & $\{8\}$ & $(8, 3)$ & $\{3\}$ & $(3, 6)$ & $(6, 12)$ & $\{12\}$ & $\{12\}$\\
				\hline 
				$i = 2$ & $(9, 13)$ & $\{9\}$ & $\{9\}$ & $(9, 4)$ & $(4, 1)$ & $\{1\}$ & $(1, 13)$ & $\{13\}$\\
				\hline
				$i = 3$ & $(15, 10)$ & $\{15\}$ & $(15, 7)$ & $\{7\}$ & $(7, 0)$ & $(0, 10)$ & $\{10\}$ & $\{10\}$\\
				\hline
				$i = 4$ & $(14, 11)$ & $\{14\}$ & $\{14\}$ & $(14, 2)$ & $(2, 5)$ & $\{5\}$ & $(5, 11)$ & $\{11\}$\\
				\hline 
			\end{tabular}
		\end{center} 
		Note that, the subgraph $G_3$ is the hypercube with $8$ vertices depicted in \ref{hypercube_with_8_vertices} along with the isolated vertices $8, 9, \dots 15$. Here, we utilize PSTs on hypercube for simultaneous quantum communication between distinct senders and receivers in nets. The other $G_j$s are subgraphs consists of paths of lengths $1$ or $2$ for making state transfer between the indicated vertices in the table.
		\begin{figure}
			\begin{subfigure}{.48\textwidth} 
				\centering 
				\begin{tikzpicture} 
					\draw [fill] (1, 1) circle [radius = .5mm];
					\node [above] at (1, 1) {$3$};
					\draw [fill] (1, -1) circle [radius = .5mm];
					\node [below] at (1, -1) {$8$};
					\draw [fill] (-1, -1) circle [radius = .5mm];
					\node [below] at (-1, -1) {$7$};
					\draw [fill] (-1, 1) circle [radius = .5mm];
					\node [above] at (-1, 1) {$2$};
					\draw (1, 1) -- (1, -1) -- (-1, -1) -- (-1, 1) -- (1, 1);
					\draw [fill] (-3, 1) circle [radius = .5mm];
					\node [above] at (-3, 1) {$1$};
					\draw (-3, 1) -- (-1, 1);
					\draw [fill] (-3, -1) circle [radius = .5mm];
					\node [below] at (-3, -1) {$6$};
					\draw (-3, -1) -- (-1, -1);
					\draw [fill] (3, 1) circle [radius = .5mm];
					\node [above] at (3, 1) {$4$};
					\draw (3, 1) -- (1, 1);
					\draw [fill] (3, -1) circle [radius = .5mm];
					\node [below] at (3, -1) {$5$};
					\draw (3, -1) -- (1, -1);
				\end{tikzpicture} 
				\caption{Consider two nets $(1, 5)$ and $(6, 4)$. There is no two edge disjoint paths from $1$ to $5$ and from $6$ to $4$. Therefore, classical routing is not possible. But, we can communicate simultaneously from $1$ to $5$ and from $6$ to $4$ using a sequence of PSTs.}
				\label{routing with C4}
			\end{subfigure}
			\begin{subfigure}{.48\textwidth}
				\centering
				\begin{tikzpicture}[scale = .8]
					\draw [fill] (1, 1) circle [radius = .5mm];
					\node [below left] at (1, 1) {$0$};
					\draw [fill] (1, -1) circle [radius = .5mm];
					\node [above left] at (1, -1) {$1$};
					\draw [fill] (-1, -1) circle [radius = .5mm];
					\node [above right] at (-1, -1) {$2$};
					\draw [fill] (-1, 1) circle [radius = .5mm];
					\node [below right] at (-1, 1) {$3$};
					\draw (1, 1) -- (1, -1) -- (-1, -1) -- (-1, 1) -- (1, 1);
					\draw [fill] (2, 2) circle [radius = .5mm];
					\node [above right] at (2, 2) {$5$};
					\draw [fill] (2, -2) circle [radius = .5mm];
					\node [below right] at (2, -2) {$6$};
					\draw [fill] (-2, -2) circle [radius = .5mm];
					\node [below left] at (-2, -2) {$7$};
					\draw [fill] (-2, 2) circle [radius = .5mm];
					\node [above left] at (-2, 2) {$4$};
					\draw (2, 2) -- (2, -2) -- (-2, -2) -- (-2, 2) -- (2, 2);
					\draw (1, 1) -- (2, 2);
					\draw (1, -1) -- (2, -2);
					\draw (-1, 1) -- (-2, 2);
					\draw (-1, -1) -- (-2, -2);
					\draw [fill] (-3, 2) circle [radius = .5mm];
					\node [left] at (-3, 2) {$8$};
					\draw (-2, 2) -- (-3, 2);
					\draw [fill] (-2, 3) circle [radius = .5mm];
					\node [above] at (-2, 3) {$9$};
					\draw (-2, 2) -- (-2, 3);
					\draw [fill] (3, 2) circle [radius = .5mm];
					\node [right] at (3, 2) {$11$};
					\draw (2, 2) -- (3, 2);
					\draw [fill] (2, 3) circle [radius = .5mm];
					\node [above] at (2, 3) {$10$};
					\draw (2, 2) -- (2, 3);
					\draw [fill] (2, -3) circle [radius = .5mm];
					\node [below] at (2, -3) {$13$};
					\draw (2, -2) -- (2, -3);
					\draw [fill] (3, -2) circle [radius = .5mm];
					\node [right] at (3, -2) {$12$};
					\draw (2, -2) -- (3, -2);
					\draw [fill] (-3, -2) circle [radius = .5mm];
					\node [left] at (-3, -2) {$15$};
					\draw (-2, -2) -- (-3, -2);
					\draw [fill] (-2, -3) circle [radius = .5mm];
					\node [below] at (-2, -3) {$14$};
					\draw (-2, -2) -- (-2, -3);
				\end{tikzpicture}
				\caption{Classical routing is not possible when the nets are $(8, 12)$, $(9, 13)$, $(15, 10)$ and $(14, 11)$. But quantum communication is possible using a sequence of PSTs.}
				\label{routing with H3}
			\end{subfigure}
			\caption{In figure \ref{routing with C4} and \ref{routing with H3} the hypercubes are the induced subgraphs generated by the vertex sets $\{2, 3, 7, 8\}$ and $\{0, 1, 2, 3, 4, 5, 6, 7\}$, respectively. A hypercube graph allows PST between its antipodal vertices. It facilitates us to perform quantum routing using all nets.}
		\end{figure}
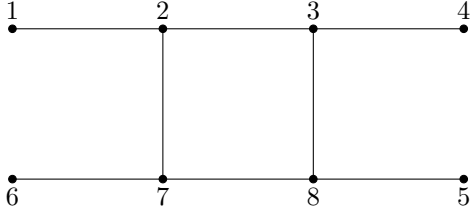
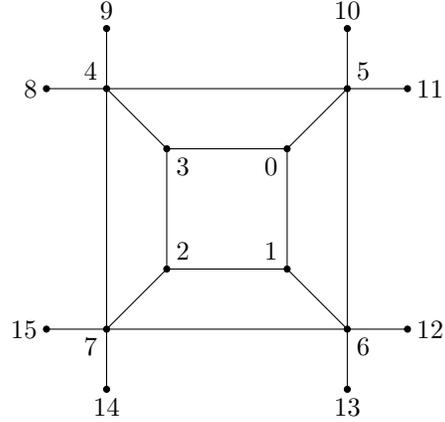

	\section{Conclusion and discussions}
		
		In this work, we develop a method for quantum communication between any two vertices of a planar graph using multiple PSTs. The foundation of multiple PSTs in a graph is Hamiltonian engineering which allows us to remove unwanted couplings between spin objects to perform a particular PST. In the language of graph theory, it is a graph switching method which switch on and off a number of edges to build up a particular subgraph allowing PST between two selected vertices. We discuss two methods to build up a scalable network where we can communicate between two arbitrarily chosen vertices using multiple PSTs. We then illustrate an idea of quantum routing on a planar graph where we keep multiple senders and receivers on the boundary of the infinite face of the graph. Our examples show that quantum routing is more efficient than their classical counterparts.
		
		Quantum routing based on PST is a well-studied topic in literature. It is proved in \cite{kay2011basics} that routing with a single state transfer between many different recipients is impossible. Therefore, multiple PST is essential for routing which we explain in this article. Performance of multiple PST needs control over the interactions. This could be experimentally challenging. In addition, it has been demonstrated that very modest controls on the interactions achieve routing with a high transfer rate \cite{kostak2007perfect, pemberton2011perfect}. In \cite{bose2009communication}, a PST based routing protocol is developed. It is proved that PST is possible between any two vertices $u$ and $v$ in the graph $K - \{(u, v)\}$, which is the complete graph excluding the particular edge $(u, v)$. If the graph $K - \{(u, v)\}$ has $n$ vertices then it has $\frac{(n+1)(n - 2)}{2}$ edges. Our procedure allows all to all communication where the graph consists of any number of edges which can be considerably less than $\frac{(n+1)(n - 2)}{2}$. Although a complete solution to the quantum routing problem in the plane will be a challenging task.

	\section*{Acknowledgment}
	
		The author acknowledges Sougato Bose and Subhashish Banerjee for some suggestions related to this work.

	\section*{Funding}
		This work is supported by the SERB funded project entitled ``Transmission of quantum information using perfect state transfer" (Grant no. CRG/2021/001834).
		
%	\bibliographystyle{unsrt}
%	\bibliography{library}

\end{document}